\documentclass[a4paper]{llncs}
\usepackage{amssymb}
\usepackage{xspace}
\usepackage{tikz}
\usepackage{epstopdf}

\newcommand{\0}{\ensuremath{\textsf 0}}
\newcommand{\1}{\ensuremath{\textsf 1}}
\newcommand{\2}{\ensuremath{\textsf 2}}

\newcommand{\x}{\ensuremath{\textsf x}}

\newcommand{\Gm}{\ensuremath{\Box}}
\newcommand{\Fm}{\ensuremath{\Diamond}}
\newcommand{\U}{\ensuremath{\textbf{U}}}
\newcommand{\UU}{\ensuremath{\textbf{UU}}}
\newcommand{\X}{\ensuremath{\textbf{X}}}

\newcommand{\nfa}{\ensuremath{NFA}\xspace}
\newcommand{\dfa}{\ensuremath{DFA}\xspace}
\newcommand{\A}{\ensuremath{\mathcal{A}}\xspace}

\newcommand{\Lang}{\ensuremath{\mathcal{L}}\xspace}
\newcommand{\M}{\ensuremath{\mathcal{M}}~}
\newcommand{\W}{\ensuremath{\mathcal{W}}\xspace}

\newcommand{\bigo}[1] {\ensuremath{O(#1)\xspace}}
\newcommand{\bigomega}[1] {\ensuremath{\Omega(#1)\xspace}}
\newcommand{\focc}[1] {\ensuremath{\Xi(#1)\xspace}}

\newcommand{\init} {\ensuremath{\mbox{\textbf{init}}(D_1:\Theta_1~/~ D_2:\Theta_2)}}
\newcommand{\triggers} {\ensuremath{\mbox{\textbf{triggers}}(D_1:\Theta_1\leadsto D_2:\Theta_2/D_3:\Theta_3)}}
\newcommand{\imp} {\ensuremath{\mbox{\textbf{implies}}(D_1:\Theta_1\leadsto D_2:\Theta_2)}}
\newcommand{\follows} {\ensuremath{\mbox{\textbf{follows}}(D_1:\Theta_1\leadsto D_2:\Theta_2/D_3:\Theta_3)}}
\newcommand{\pref}[1] {\ensuremath{\mbox{\textbf{pref}}(#1)~}}

\newcommand{\anti}[1] {\ensuremath{\mbox{\textbf{anti}}(#1)}}

\newcommand{\lags}{\ensuremath{\textbf{lags}}}
\newcommand{\tracks} {\ensuremath{\textbf{tracks}}}
\newcommand{\sep}{\ensuremath{\textbf{sep}}}
\newcommand{\ubound}{\ensuremath{\textbf{ubound}}}

\newcommand{\wf} {\ensuremath{\textsf{Wf~}}}
\newcommand{\modelsv} {\ensuremath{\models_\nu}}

\newcommand {\And} {\ensuremath{\wedge}}
\newcommand {\Or} {\ensuremath{\vee}}
\newcommand {\union} {\ensuremath{\cup}}

\newcommand {\implies} {\ensuremath{\Rightarrow}}
\renewcommand {\iff} {\ensuremath{\Leftrightarrow}}

\newcommand {\qddc} {QDDC\xspace}

\newcommand {\ce} {Ce\xspace}
\newcommand {\sece} {SeCe\xspace}
\newcommand {\secenl} {SeCeNL\xspace}
\newcommand{\sececnt}{SeCe\xspace}
\newcommand{\sececntnl}{SeCeNL\xspace}

\newcommand {\ltl} {LTL\xspace}

\newcommand {\psl} {PSL-Sugar\xspace}
\newcommand {\mtl} {MTL\xspace}

\newcommand {\ang}[1] {\ensuremath{\langle#1\rangle}}

\newcommand {\sq}[1] {\ensuremath{[#1]}}
\newcommand {\dsq}[1] {\ensuremath{[[#1]]}}
\newcommand {\dcurly}[1] {\ensuremath{\{\{#1\}\}}}
\newcommand {\defeq} {\ensuremath{\stackrel{\mathrm{def}}{\equiv}}}

\newcommand {\len}[1] {\ensuremath{len(#1)}}

\newcommand {\intv}[1] {\ensuremath{Intv(#1)}}
\newcommand {\slen} {\ensuremath{slen}\xspace}
\newcommand {\scount} {\ensuremath{scount}\xspace}
\newcommand {\sdur} {\ensuremath{sdur}\xspace}
\newcommand {\pt} {\ensuremath{pt}~}
\newcommand {\ext} {\ensuremath{ext}~}
\newcommand {\true} {\ensuremath{true}~}

\newcommand {\dom}[1] {\ensuremath{dom(#1)}}

\newcommand {\nat} {\ensuremath{\mathbb{N}}}

\newcommand {\eqv}[1] {\ensuremath{[[#1]]}}

\newcommand{\oomit}[1]{}

\begin{document}
\title{Formalizing Timing Diagram Requirements in Discrete Duration Calulus}
\titlerunning{Hamiltonian Mechanics}  
\author{Raj Mohan Matteplackel\inst{1} \and Paritosh K.~Pandya\inst{1} \and Amol Wakankar\inst{2}}
\institute{Tata Institute of Fundamental Research, Mumbai 400005, India.\\
\email{\{raj.matteplackel,pandya\}@tifr.res.in}
\and
Bhabha Atomic Research Centre, Mumbai, India.\\
\email{amolk@barc.gov.in}}
\maketitle              

\begin{abstract}
Several temporal logics have been proposed to formalise \emph{timing diagram} requirements
over hardware and embedded controllers. These include LTL \cite{CF05}, 
discrete time MTL \cite{AH93} and the recent industry standard 
PSL \cite{EF16}. However, succintness and visual structure of a timing diagram are not adequately captured by their formulae 
\cite{CF05}.
Interval temporal logic QDDC is a highly succint and visual notation for specifying patterns of behaviours
\cite{Pan00}.  

In this paper, we propose a practically useful notation called \sececntnl 
which enhances negation free fragment of \qddc with features of \emph{nominals} and \emph{limited liveness}. 
We show that timing diagrams  can be naturally (compositionally)
and succintly  formalized in  \sececntnl as compared with \psl and \mtl. We give a linear time translation from timing diagrams to
\sececntnl. As our second main result, we  propose a linear time 
translation of \sececntnl into \qddc. 
This allows QDDC tools such as DCVALID \cite{Pan00,Pan01} and DCSynth to be used
for checking consistency of timing diagram requirements as well as for automatic synthesis of property monitors and controllers.
We give examples of a minepump controller and a bus arbiter to illustrate our tools. Giving a theoretical analysis, we show that for the proposed \sececntnl, the satisfiability and model checking have elementary complexity as compared to the non-elementary 
complexity for the full logic \qddc.


\end{abstract}

\section{Introduction}
\label{section:intro}
A \emph{timing diagram} is a collection of binary signals 
and a set of timing constraints on them. 
It is a widely used visual formalism in the realm of digital hardware design,
communication protocol specification and
embedded controller specification. 
The advantages of timing diagrams in hardware design are twofold, 
one, since designers can visualize waveforms of signals 
they are easy to comprehend and two, they are very convenient 
for specifying ordering and timing constraints between events 
(see figures Fig.~\ref{fig:intro} and Fig.~\ref{fig:lags} below).

There have been numerous attempts at formalizing timing diagram constraints
in the framework of temporal logics such as the
\emph{timing diagram logic} \cite{Fis99}, with
\emph{LTL formulas} \cite{CF05}, and as \emph{synchronous regular timing diagrams} 
\cite{AEKN00}. Moreover, there are  industry standard 
property specification languages 
such as  \emph{PSL/Sugar} and  \emph{OVA} 
for associating temporal assertions to hardware designs \cite{EF16}. 
The main motivation for these attempts was to 
exploit automatic verification techniques 
that these formalisms support for validation and automatic circuit synthesis.
However, commenting on their success, Fisler et.~al.~state that 
the less than satisfactory adoption of formal methods 
in timing diagram domain can be partly attributed 
to the gulf that exists between 
graphical timing diagrams and textual temporal logic -- 
expressing various timing dependencies that can exist among signals 
that can be illustrated so naturally in timing diagrams 
is rather tedious in temporal logics \cite{CF05}. 
As a result, hardware designers use 
timing diagrams informally without any well defined 
semantics which make them unamenable to 
automatic design verification techniques. 

In this paper, we take a fresh look at formalizing 
timing diagram requirements with emphasis 
on the following three features of the formalism that we propose here.

Firstly, we propose the use of an \emph{interval temporal logic} \qddc 
to specify patterns of behaviours.  
\qddc is a highly succinct and visual notation for specifying regular
patterns of behaviours \cite{Pan00,Pan01,KP05}. 
We identify a quantifier and negation-free subset \sececnt  of \qddc which is
sufficient for formalizing timing diagram patterns.  
It includes generalized regular expression like syntax with counting constructs.
Constraints imposed by timing diagrams are straightforwardly 
and compactly stated in this logic. 
For example, the timing diagram
in Fig.~\ref{fig:intro} stating that $P$ transits from 0 to 1 somewhere in interval
$u$ to $u+3$ cycles is captured by the \sececnt formula 
\verb|[|$\neg$ \verb|P]^<u>^(slen=3| $\And$\verb| [|$\neg$\verb|P]^[[P]])^[[P]]|.
\emph{The main advantage of \sececnt\/ is that it has elementary satisfiability as compared to the non-elementary
satisfiability of general \qddc}.
\begin{figure}[!h]
\begin{center}
\includegraphics[height=1cm, keepaspectratio]{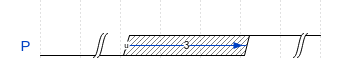}
\end{center}
\caption{Timing diagram with a marked position $u$ and a timing constraint.}
\label{fig:intro}
\end{figure}

Secondly, it is very typical for timing diagrams to have partial ordering 
and synchronization constraints between distinct events. 
Emphasizing this aspect, 
formalisms such as two dimensional regular expressions \cite{Fis07} 
have been proposed for timing diagrams. 
We find that synchronization in timing diagram may even extend 
across different patterns of limited liveness properties. 
In order to handle such synchronization, 
we extend our logic \sececnt\/  with \emph{nominals} from hybrid temporal logics \cite{FRS03}. 
Nominals are temporal variables which ``freeze'' the positions of 
occurrences of events. They naturally allow synchronization across formulae.

Thirdly, we enhance the timing diagram  specifications (as well as logic \sececnt)
with {\em limited liveness operators}. While
timing diagrams visually specify patterns  of occurrence of signals, they do not make precise the modalities
of occurrences of such patterns. 
We explicitly introduce modalities such as 
\textit{a}) initially, a specified pattern must occur, or
that \textit{b}) 
every occurrence of \emph{pattern1} 
is necessarily and immediately followed by an occurrence of \emph{pattern2}, 
or that \textit{c}) occurrence of a specified pattern is 
forbidden anywhere within a behaviour.
In this, we are inspired by Allen's Interval Algebra relations \cite{All83} as well as the LSC operators of Harel for 
message sequence charts \cite{DH01}. We confine ourselves to 
{\em limited liveness properties} where good things are achieved within
specified bounds.
For example, in specifying a modulo 6 counter, we can say that 
the counter will stabilize  before completion of first 15 cycles. 
Astute readers will notice that, technically, 
our limited liveness operators only give rise 
to ``safety'' properties (in the sense of Alpern and 
Schneider \cite{AS87}). However, from a designer's perspective 
they do achieve the practical goal of forcing good things to happen. 

Putting all these together, we define a logic \sececntnl which includes negation-free \qddc together
with   limited liveness operators  as well as nominals. 
The formal syntax and semantics of \sececntnl formulas is given in \S\ref{subsection:secenl}. 
We claim that \sececntnl provides a natural and convenient formalism for encoding timing diagram requirements. 
{\em Substantiating this, we formulate a translation of timing diagrams into \secenl formulae in \S\ref{section:td}}. 
The translation is succinct, in fact, linear time computable in the size of the timing diagram.
(A textual syntax is used for timing diagrams. The textual syntax of timing  diagrams used is inspired by the tool 
WaveDrom \cite{CP16},  which is also used for graphical rendering of our timing diagram specifications.)
Moreover, the translation is compositional, i.e. it translates each element of the timing diagram as one small formula and 
overall specification is just the conjunction of such constraints. Hence, the translation preserves the structure of the diagram.

With several examples of timing diagrams, we compare its \sececntnl formula with the formula in logics such as \psl and \mtl.
Logic \psl is amongst the most expressive notations for requirements. Logic \psl is syntactically a superset of \mtl and \ltl. It extends \ltl with SERE (regular expressions with intersection)
which are similar to our \sececnt. In spite of this, we a show natural examples where \sececntnl formula is at least one exponent more succinct as compared to \psl. 

As the second main contribution of this paper,  we consider formal verification and controller synthesis from \sececntnl specifications.
In \S\ref{subsection:td2secenl}, we formulate a reduction from a \sececntnl formula to an equivalent \qddc formula. 
This allows \qddc tools to be used for \sececntnl. It may be noted that, though expressively no more powerful than \qddc, 
logic \sececntnl considerably more efficient for satisfiability and model checking. We show that these problems have elementary complexity
as compared with full \qddc which exhibits non-elementary complexity. Also, the presence of limited liveness and nominals makes it more
convenient as compared to \qddc for practical use.

By implementing the above reductions, we have constructed a Python based translator  which converts a requirement 
consisting of a boolean combination  of timing diagram 
specifications (augmented with limited liveness) and 
\sececntnl formulae into an equivalent \qddc formula. 
We can analyze the resulting formula  using the \qddc tools DCVALID \cite{Pan00,Pan01} as well as DCSynthG 
for model checking and controller synthesis, respectively  (see Fig.~\ref{fig:tool-chain} for the tool chain). 
We illustrate the use of our tools by the case studies of 
a synchronous bus arbiter and a minepump controller in \S\ref{section:casestudy}. 
Readers may note that we specify rather rich 
quantitative requirements not commonly considered, 
and our tools are able to automatically synthesize monitors and controllers for 
such specifications.


\section{Logic \qddc}
\label{section:qddc}
Let $\Sigma$ be a finite non empty set of propositional variables. 
A \emph{word} $\sigma$ over $\Sigma$ is a finite sequence of the form 
$P_0\cdots P_n$ where $P_i\subseteq\Sigma$ for each $i\in\{0,\ldots,n\}$. 
Let $\len{\sigma}=n+1$, 
$\dom{\sigma}=\{0,\ldots,n\}$ and $\forall i\in\dom{\sigma}:\sigma(i)=P_i$.  

The syntax of a \emph{propositional formula} over $\Sigma$ is given by:
\[
\varphi := \0\ |\ \1\ |\ p\in\Sigma\ |\ \varphi\And\varphi\ |\ \varphi\Or\varphi\ |\ \neg\varphi, 
\]
and operators such as $\implies$ and $\iff$ are defined as usual. 
Let $\Omega_\Sigma$ be the set of all propositional formulas over $\Sigma$. 

Let $\sigma=P_0\cdots P_n$ be a word and $\varphi\in\Omega_\Sigma$. 
Then, for an $i\in\dom{\sigma}$ 
the satisfaction relation $\sigma,i\models\varphi$ is defined inductively as expected:
$\sigma, i \models \1$; $\sigma, i \models p$ iff $p\in\sigma(i)$; 
$\sigma, i \models \neg p$ iff $\sigma,i\not\models p$, 
and the satisfaction relation for the rest of the 
boolean combinations defined in a natural way.

The syntax of a \qddc formula over $\Sigma$ is given by: 
\[
\begin{array}{lc}
D:= &\ang{\varphi}\ |\ \sq{\varphi}\ |\ \dsq{\varphi}\ |\ \dcurly{\varphi}\ |\ D\ \verb|^|\ D\ |\ 
\neg D\ |\ D\Or D\ |\ D\And D\ |\ D^*\ |\ \vspace{1mm}\\
&\exists p.\ D\ |\ \forall p.\ D\ |\ 
slen \bowtie c\ |\ scount\ \varphi \bowtie c\ |\ sdur\ \varphi \bowtie c, 
\end{array} 
\]
where $\varphi\in\Omega_\Sigma$, $p\in\Sigma$, 
$c\in\nat$ and $\bowtie\in\{<,\leq,=,\geq,>\}$. 

An \emph{interval} over a word $\sigma$ is of the form $[b,e]$ 
where $b,e\in\dom{\sigma}$ and $b\leq e$. 
An interval $[b_1,e_1]$ is a sub interval of $[b,e]$ 
if $b\leq b_1$ and $e_1\leq e$. 
Let $\intv{\sigma}$ be the set of all intervals over $\sigma$.

Let $\sigma$ be a word over $\Sigma$ and let $[b,e]\in\intv{\sigma}$ be an interval. 
Then the satisfaction relation of a \qddc formula $D$ over 
$\Sigma$, written $\sigma,[b,e]\models D$, is defined inductively as follows:
\[
\begin{array}{lcl}
\sigma, [b,e]\models\ang{\varphi} & \mathrm{\ iff \ } & \sigma,b\models \varphi,\\
\sigma, [b,e]\models\sq{\varphi} & \mathrm{\ iff \ } & \forall b\leq i<e:\sigma,i\models \varphi,\\
\sigma, [b,e]\models\dsq{\varphi} & \mathrm{\ iff \ } & \forall b\leq i\leq e:\sigma,i\models \varphi,\\
\sigma, [b,e]\models\dcurly{\varphi} & \mathrm{\ iff \ } & e=b+1 \mbox{ and }\sigma,b\models \varphi,\\
\sigma, [b,e]\models\neg D & \mathrm{\ iff \ } & \sigma,[b,e]\not\models D,\\
\sigma, [b,e]\models D_1\Or D_2 & \mathrm{\ iff \ } & \sigma, [b,e]\models D_1\mbox{ or }\sigma,[b,e]\models D_2,\\
\sigma, [b,e]\models D_1\And D_2 & \mathrm{\ iff \ } & \sigma, [b,e]\models D_1\mbox{ and }\sigma,[b,e]\models D_2,\\
\sigma, [b,e]\models D_1\verb|^| D_2 & \mathrm{\ iff \ } & \exists b\leq i\leq e:\sigma, [b,i]\models D_1\mbox{ and }\sigma,[i,e]\models D_2.\\
\end{array}
\]
We call word $\sigma'$ a $p$-variant, $p\in\Sigma$, of a word $\sigma$ 
if $\forall i\in\dom{\sigma},\forall q\neq p:\sigma'(i)(q)=\sigma(i)(q)$. 
Then $\sigma,[b,e]\models\exists p.~D\iff\sigma',[b,e]\models D$ for some 
$p$-variant $\sigma'$ of $\sigma$ and, 
$\sigma,[b,e]\models\forall p.~D\iff\sigma,[b,e]\not\models\exists p.~\neg D$. 
We define $\sigma\models D$ iff $\sigma,[0,\len{\sigma}]\models D$.

\begin{example}
Let $\Sigma=\{p,q\}$ and let 
$\sigma=P_0\cdots P_7$ be such that $\forall 0\leq i<7:P_i=\{p\}$ and $P_7=\{q\}$. 
Then $\sigma, [0,7]\models\sq{p}$ but not $\sigma, [0,7]\models\dsq{p}$ 
as $p\not\in P_7$.   
\end{example}

\begin{example}
Let $\Sigma=\{p,q,r\}$ and let 
$\sigma=P_0\cdots P_{10}$ be such that $\forall 0\leq i<4:P_i=\{p\}$, 
$\forall 4\leq i<8:P_i=\{p,q,r\}$ and $\forall 8\leq i\leq 10:P_i=\{q,r\}$.  
Then 
\[
\sigma, [0,10]\models \sq{p}\verb|^|\dsq{\neg p\And r}
\]
because for $i\in\{8,9,10\}$ 
the condition $\exists 0\leq i\leq 10:\sigma, [0,i]\models\sq{p}\mbox{ and }\sigma,[i,10]\models\dsq{\neg p\And r}$ 
is met. But $\sigma, [0,7]\not\models\sq{p}\verb|^|\dsq{\neg p\And r}$ as 
$\neg\exists 0\leq i\leq 7:\sigma, [0,i]\models\sq{p}\mbox{ and }\sigma,[i,7]\models\dsq{\neg p\And r}$. 
\end{example}

Entities \slen, \scount, and \sdur are called \emph{terms} in \qddc. 
The term \slen gives the length of the interval in which it is 
measured, $\scount\ \varphi$ where $\varphi\in\Omega_\Sigma$, counts 
the number of positions including the last point 
in the interval under consideration where $\varphi$ holds, and    
$\sdur\ \varphi$ gives the number of positions excluding the last point 
in the interval where $\varphi$ holds. 
Formally, for $\varphi\in\Omega_\Sigma$ we have 
$\slen(\sigma, [b,e])=e-b$, 
$\scount(\sigma,\varphi,[b,e])=\sum_{i=b}^{i=e}\left\{\begin{array}{ll}
					1,&\mbox{if }\sigma,i\models\varphi,\\
					0,&\mbox{otherwise.}
					\end{array}\right\}$ and 
$\sdur(\sigma,\varphi,[b,e])=\sum_{i=b}^{i=e-1}\left\{\begin{array}{ll}
					1,&\mbox{if }\sigma,i\models\varphi,\\
					0,&\mbox{otherwise.}
					\end{array}\right\}$
In addition we also use the following derived constructs: 
$\sigma, [b,e]\models\pt$ iff $b=e$; $\sigma, [b,e]\models\ext$ iff $b<e$; 
$\sigma, [b,e]\models \Fm D$ iff $\true\verb|^|D\verb|^|\true$ and 
$\sigma, [b,e]\models \Gm D$ iff $\sigma, [b,e]\not\models\Fm\neg D$. 

A \emph{formula automaton} for a \qddc formula $D$ is a 
\emph{deterministic finite state automaton} which accepts precisely language 
$L=\{\sigma\ |\ \sigma\models D\}$.

\begin{theorem}
\label{theorem:formula-automaton}\cite{Pan01}
For every \qddc formula $D$ over $\Sigma$ we can construct a \dfa $\A(D)$ 
for $D$ such $\Lang(D)=\Lang(\A(D))$. 
The size of $\A(D)$ is non elementary in the size of $D$ in the worst case.
\end{theorem}

\subsection{Chop expressions: \ce and \sece}
\label{subsection:sece}
\begin{definition}
The logic \emph{Semi extended Chop expressions} (\sece) is a syntactic subset of \qddc in which 
the operators $\exists p.~D$, $\forall p.~D$ and negation are not allowed.
The logic \emph{Chop expressions} (\ce) is a sublogic of 
\sece in which conjuction is not allowed.
\end{definition}

\begin{lemma}
\label{lemma:ce-size}
For any chop expression $D$ of size $n$ we can effectively construct a language equivalent 
\dfa \A of size $\bigomega{\ensuremath{2^{2^n}}}$. 
\end{lemma}  

\begin{proof}
We observe that for any chop expression $D$ we can construct 
a language equivalent \nfa which is at most exponential in size of $D$ 
including the constants appearing in it 
(for a detailed proof see \cite{BP12} wherein a similar result has been proved). 
But this implies there exists a \dfa of size $2^{2^n}$ 
which accepts exactly the set of 
words $\sigma$ such that $\sigma\models D$. 
\end{proof}

\begin{corollary}
For any \sece $D$ of size $n$ we can effectively construct a language equivalent 
\dfa \A of size $\bigomega{\ensuremath{2^{2^{2^n}}}}$. 
\end{corollary}
\begin{proof}
Proof follows from the definition of \sece, lemma~\ref{lemma:ce-size} 
and from the fact that 
the size of the product of \dfa's can be atmost exponential in the size 
of individual \dfa's.
\end{proof}

\subsection{DCVALID and DCSynthG}
\label{subsection:dcvalid}
%

The reduction from a \qddc formula to its formula automaton 
has been implemented into the tool \emph{DCVALID} \cite{Pan00,Pan01}. 
The formula automaton it generates is total, deterministic 
and minimal automaton for the formula. DCVALID can also translate 
the formula automaton into Lustre/SCADE, Esterel, SMV and Verilog 
\emph{observer module}. By connecting this observer module to run synchronously with
a system we can reduce model checking of QDDC property to reachability checking
in observer augmented system. See \cite{Pan00,Pan01} for details.
A further use of formula automata can be seen in 
the tool called \emph{DCSynthG} which synthesizes synchronous 
dataflow controller in SCADE/NuSMV/Verilog from \qddc specification. 


\subsection{Logic \secenl: Syntax and Semantics}
\label{subsection:secenl}

We can now introduce our logic \secenl which builds upon \sece 
by augmenting it with 
\emph{nominals} and \emph{limited liveness operators}. 
\paragraph{Syntax}:
The syntax of \secenl atomic formula is as follows. 
Let $D$, $D_1$, $D_2$ and $D_3$ range over \sece formulae
and let $\Theta$, $\Theta_1$, $\Theta_2$ and $\Theta_3$ range over subset of 
propositional variables occurring in \sece formula. 
The notation $D:\Theta$, called a \emph{nominated formula}, denotes that $\Theta$ 
is the set of variables used as nominals in the formula $D$.
\[
\begin{array}{l}
   \init ~\mid~ \anti{D:\Theta} ~\mid~ \pref{D:\Theta} ~\mid\\ 
    \imp ~\mid~ \follows ~~\mid \\
    \triggers
\end{array} 
\]
{\em An \secenl formula is a boolean combination of atomic \secenl formulae of the form above.}
As a convention, $D:\{ \}$ is abbreviated as $D$ when the set of nominals $\Theta$ is empty.


\subsubsection{Limited Liveness Operators}: 
Given an word $\sigma$ and a position $i \in\dom{\sigma}$, 
we state that $\sigma,i \models D$ iff $\sigma[0:i] \models D$.
Thus, the interpretation is that the past of the position $i$ in execution satisfies $D$.

For a \sece formula $D$ we let $\focc{D}=D\And\neg (D\verb|^|ext)$, 
which says that if $\sigma,[b,e]\models\focc{D}$ then 
$\sigma,[b,e]\models D$ and 
there exists no proper prefix interval $[b,e_1]$, 
(i.~e.~$[b,e_1]\in\intv{\sigma}$ and $b \leq e_1 < e$) such that 
$\sigma,[b,e_1]\models D$. 
We say $\sigma'\leq_{prefix}\sigma$ if $\sigma'$ is a prefix of $\sigma$, and 
$\sigma' <_{prefix} \sigma$ if $\sigma'$ is a proper prefix of $\sigma$. 

We first explain the semantics of {\em limited liveness operators} 
assuming that no nominals are used in the specification, 
i.e.~$\Theta$, $\Theta_1$, $\Theta_2$ and $\Theta_3$ are all empty. 
A set $S\subseteq\Sigma^*$ is \emph{prefix closed} if $\sigma\in S$ 
then $\forall\sigma':\sigma'\leq_{prefix}\sigma\implies\sigma'\in S$.
We observe that each atomic liveness formula 
denotes a prefix closed subset of $(2^\Sigma)^+$.
\begin{itemize}
 \item $L(\pref{D})=\{\sigma\ |\ \forall\sigma'\leq_{prefix}\sigma:\sigma'\models D\}.$
 Operator $\pref{D}$ denotes that $D$  holds invariantly throughout the execution. 
 \item $L(\mbox{\textbf{init}}(D_1/ D_2))=
         \{\sigma\ |\ \forall j:\sigma,[0,j]\models D_2\implies\exists k\leq j:\sigma,[0,k]\models D_1\}.$
 Operator $\mbox{\textbf{init}}(D_1/ D_2)$ basically 
states that if $j$ is the first position which satisfies 
 $D_2$ in the execution then there exists an $i  \leq j$ such that $i$ satisfies $D_1$. Thus, initially $D_1$ holds before $D_2$
 unless the execution (is too short and hence) does not satisfy $D_2$ anywhere.
 \item $L(\mbox{\textbf{anti}}(D))=\{\sigma\ |\ \forall i,j:\sigma,[i,j]\not\models D\}$. 
Operator $\mbox{\textbf{anti}}(D)$ states that there is no observation sub interval of the execution which satisfies $D$.
 \item $L(\mbox{\textbf{implies}}(D_1\leadsto D_2))=
         \{\sigma\ |\ \forall i,j:(\sigma,[i,j]\models D_1\implies\sigma,[i,j]\models D_2)\}$. 
Operator $\mbox{\textbf{implies}}(D_1\leadsto D_2)$ states  all observation intervals which satisfy $D_1$ will also satisfy
 $D_2$.
 \item $L(\mbox{\textbf{follows}}(D_1\leadsto D_2/D_3))=
                 \{\sigma\ |\ \forall i,j:(\sigma,[i,j]\models D_1\implies$ \\
        \hspace*{37mm} $(\forall k:\sigma,[j,k]\models\focc{D_3}\implies\exists l\leq k:\sigma,[j,l]\models D_2))\}.$ \\
 Operator $\mbox{\textbf{follows}}(D_1\leadsto D_2/D_3)$ 
states that if any observation interval $[i,j]$ satisfies $D_1$
 and there is a following shortest interval $[j,k]$ which satisfies $D_3$ then 
there exists a prefix interval of $[j,k]$ which satisfies $D_2$.
 
 \item $L(\mbox{\textbf{triggers}}(D_1\leadsto D_2/D_3))= 
                   \{\sigma\ |\ \forall i,j:(\sigma,[i,j]\models D_1\implies $ \\
   \hspace*{38mm} $(\forall k:\sigma,[i,k]\models\focc{D_3}\implies\exists l\leq k:\sigma,[i,l]\models D_2))\}.$ \\
  Operator $\mbox{\textbf{triggers}}(D_1\leadsto D_2/D_3)$ states 
that if any observation interval $[i,j]$ satisfies $D_1$ and
  if $[i,k]$ is the shortest interval which satisfies $D_3$ 
then $D_2$ holds for a prefix interval of $[i,k]$. 
\end{itemize}

Based on this semantics, we can translate an atomic \secenl formula 
$\zeta$ without nominals into equivalent $\sece$ formula 
$\aleph(\zeta)$ as follows.
\begin{enumerate}
\item $\aleph(\pref{D})\defeq\neg((\neg D)\verb|^|\true)$.
\item $\aleph(\mbox{\textbf{init}}(D_1/ D_2))\defeq
pref(\focc{D_2}\Rightarrow D_1\verb|^|\true)$.
\item $\aleph(\mbox{\textbf{anti}}(D))\defeq \neg (\true\verb|^|D\verb|^|\true)$.
\item $\aleph(\mbox{\textbf{implies}}(D_1\leadsto D_2))\defeq
		\Box(D_1\implies D_2)$.
\item $\aleph(\mbox{\textbf{follows}}(D_1\leadsto D_2/D_3))\defeq\Box(\neg (D_1\verb|^|(\focc{D_3}\And\neg(D_2\verb|^|\true))))$.
\item $\aleph(\mbox{\textbf{triggers}}(D_1\leadsto D_2/D_3))\defeq\Box(D_1\verb|^|\true\implies(\focc{D_3}\implies D_2\verb|^|\true))\bigwedge
\Box(D_1\implies pref(\focc{D_3}\implies D_2\verb|^|\true))$.
\end{enumerate}

\begin{lemma}
For any $\zeta \in \secenl$, 
if $\zeta$ does not use nominals then 
$\sigma \in L(\zeta)$ iff $\sigma \in L(\aleph(\zeta))$. 
\end{lemma}
The proof follows from examination of the semantics of $\zeta$ 
and the definition of $\aleph(\zeta)$. We omit the details.

\subsubsection{Nominals}:
Consider a nominated formula $D:\Theta$ where $D$ is a \sece formula 
over propositional variables $\Sigma\union\Theta$.
As we shall see later, 
the propositional variables in $\Theta$ are treated as ``place holders'' 
- variables which are meant to be true exactly at one point - 
and we call them \emph{nominals} following \cite{FRS03}. 

Given an interval $[b,e]\in\intv{\nat}$ we define 
a \emph{nominal valuation} over $[b,e]$ to be a map 
$\nu : \Theta \rightarrow\{i\ |\ b\leq i\leq e\}$. 
It assigns a unique position within $[b,e]$ to each nominal variable. 
We can then straightforwardly define $\sigma,[b,e] \models_\nu D$ 
by constructing a word $\sigma_\nu$ over $\Sigma\union\Theta$ 
such that $\forall p\in\Sigma:p\in\sigma_\nu(i)\iff p\in\sigma(i)$ and 
$\forall u\in\Theta:u\in\sigma_\nu(i)\iff\nu(u)=i$. 
Then $\sigma_\nu,[b,e]\models D\iff\sigma,[b,e]\modelsv D$. 
We state that $\nu_1$ over $\Theta_1$ and $\nu_2$ over $\Theta_2$ 
are consistent if $\nu_1(u) = \nu_2(u)$ 
for all $u \in \Theta_1 \cap \Theta_2$. 
We denote this by $\nu_1 \parallel \nu_2$.

\paragraph{Semantics of \secenl}: 
Now we consider the semantics of \secenl 
where nominals are used and shared between different parts 
$D_1$, $D_2$ and $D_3$ of an atomic formula such as \imp. 
\begin{example}[lags]
\label{exam:lags}
Let $D_1:\{u,v\}$ be the formula \verb| (<u> ^ [[P]]| $\And$ \verb|((slen=n)| 
\verb|^ <v> ^ true)| which holds for an interval where $P$ is \true 
throughout the interval and $v$ marks the $n+1$ position from $u$ 
denoting the start of the interval. 
Let $D_2:\{v\}$ be the formula \verb| true ^ <v> ^ [[Q]]|.
Then, \textbf{implies}$(D_1:\{u,v\} \leadsto D_2:\{v\})$ 
states that for all observation intervals $[i,j]$ and all nominal valuations 
$\nu$ over $[i,j]$ if $\sigma,[i,j] \models_\nu D_1$ then $\sigma,[i,j] \models_\nu D_2$. 
This formula is given by \emph{live timing diagram} in Fig.~\ref{fig:lags} below. 
\begin{figure}
\centering
\includegraphics[height=25mm, keepaspectratio]{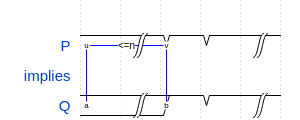}
\caption{Live timing diagram.}
\label{fig:lags}
\end{figure}
\footnote{Here we wish to point out that the illustration was made with 
the timing diagram editor WaveDrom and 
due to its limitation on naming nominals 
we were forced to rename the nominal $v$ appearing 
in $D_2$ as $a$.}
\end{example}

We now give the semantics of \secenl.
\begin{itemize} 
 \item $L(\pref{D:\Theta})=\{\sigma\ |\ \forall\sigma'\leq_{prefix}\sigma: \exists \nu .~~ \sigma'\models_\nu D\}.$
 \item $L(\init)=
         \{\sigma ~\mid~ \forall j \forall \nu:\sigma,[0,j]\models_\nu D_2\implies
               \exists k\leq j \exists \nu_2: \nu_1 \parallel \nu_2 \land \sigma,[0,k]\models_{\nu_2} D_1\}.$
 \item $L(\anti{D:\Theta})=\{\sigma\ |\ \forall i,j \forall \nu: \sigma,[i,j]\not\models_\nu D\}.$
 
 \item $L(\imp)=
         \{\sigma\ |\ \forall i,j \forall \nu_1:(\sigma,[i,j]\models_{\nu_1} D_1\implies \exists \nu_2: \nu_1 \parallel \nu_2 \land
                \sigma,[i,j]\models_{\nu_2} D_2)\}.$
 
 \item $L(\follows)=$ \\
         $\{\sigma\ |\ \forall i,j  \forall \nu_1:(\sigma,[i,j]\models_{\nu_1} D_1 \implies
         (\forall k \forall \nu_2 \parallel \nu_1: \sigma,[j,k]\models_{\nu_2} \focc{D_3} \implies \exists l\leq k \exists \nu_3:
         \nu_3 \parallel\nu_1 \land \nu_3 \parallel \nu_2 \land
         \sigma,[j,l]\models_{\nu_3} D_2))\}.$
 
 \item $L(\triggers)=$ \\
 \hspace*{1cm}    $\{\sigma\ |\ \forall i,j \forall \nu_1 :(\sigma,[i,j]\models_{\nu_1} D_1\implies
            (\forall k \forall \nu_2 \parallel \nu_1: \sigma,[i,k]\models_{\nu_2} \focc{D_3}
              \implies\exists l\leq k\exists \nu_3: \nu_3 \parallel\nu_1 \land \nu_3 \parallel \nu_2 \land \sigma,[i,l]\models D_2))\}.$
\end{itemize}

Based on the above semantics, 
we now formulate a \qddc formula equivalent 
to a \secenl formula. 
We define the following useful notations $\forall^1_\Theta:D$ 
and $\exists^1_\Theta:D$ as derived operators. 
These operators are essentially relativize quantifiers 
to restrict variables to singletons. 
\[
\begin{array}{c}
\forall^1_\Theta:D\iff
\forall u_1.\ \cdots\ .\forall u_n((scount\ u_1=1\ \wedge \cdots \wedge\ scount\ u_n=1)\ \implies\ D).\\ 
\exists^1_\Theta:D\iff
\exists u_1.\ \cdots\ .\exists u_n((scount\ u_1=1\ \wedge \cdots \wedge\ scount\ u_n=1)\ \wedge\ D).
\end{array}
\]

\paragraph{From \secenl to \qddc}:
We now define the translation $\aleph$ from \secenl to \qddc. 
\begin{enumerate}
\item $\aleph(\pref{D:\Theta})\defeq\neg(\exists^1_\Theta:\neg D\verb|^|\true)$.
\item $\aleph(\init)\defeq
pref(\forall^1_{\Theta_2}:(D_2\Rightarrow\exists^1_{\Theta_1-\Theta_2}:D_1\verb|^|\true))$.
\item $\aleph(\neg\exists D:\Theta)\defeq \neg (\exists^1_\Theta:\true\verb|^|D\verb|^|\true)$.
\item $\aleph(\imp)\defeq
		\Box(\forall^1_{\Theta_1}:(D_1\implies\exists^1_{\Theta_2-\Theta_1}:D_2))$.
\item $\aleph(\follows)\defeq\\
	\Box(\forall^1_{\Theta_1}:\forall^1_{\Theta_3-\Theta_1}:\exists^1_{\Theta_2-(\Theta_1\union\Theta_3)}:
	\neg (D_1\verb|^|(\focc{D_3}\And\neg(D_2\verb|^|\true))))$.
\item $\aleph(\triggers)\defeq\\
\Box(\forall^1_{\Theta_1}:(D_1\verb|^|\true\implies(\forall^1_{\Theta_3-\Theta_1}:(\focc{D_3}\implies\exists_{\Theta_2-(\Theta_1\union\Theta_3)}^1:D_2\verb|^|\true))))\bigwedge \\
\Box(\forall^1_{\Theta_1}:(D_1\implies pref(\forall^1_{\Theta_3-\Theta_1}:(\focc{D_3}\implies\exists_{\Theta_2-(\Theta_1\union\Theta_3)}^1:D_2\verb|^|\true))))$.
\end{enumerate}
\begin{theorem}
For any word $\sigma$ over $\Sigma$ and any $\zeta \in \secenl$ 
we have that $\sigma \in L(\zeta)$ iff $\sigma \in L(\aleph(\zeta))$. 
Moreover, the translation $\aleph(\zeta)$ can be computed 
in time linear in the size of $\zeta$. 
\end{theorem}
The proof follows from the semantics of $\zeta$ and the definition of $\aleph(\zeta)$. 

\begin{lemma}
Let $\zeta=\imp$ and let $|\A(D_i)|=m_i$ for $i\in\{1,2\}$. 
Then there exists a \dfa $\A(\zeta)$ of size 
at most $2^{2^{m_1m_2}}$ for $\zeta$. 
\end{lemma}
\begin{proof}
The formula $\zeta$ can be written in terms of a negation and two existential quantifiers. 
Note that each application of existential quantifier will result in an \nfa 
and each time we determinize we get a \dfa which is at most exponential in the size of \nfa. 
Since that both $\A(D_1)$ and $\A(D_2)$ are \dfa's to start with, 
this implies we can construct a \dfa $\A(\zeta)$ of size at most $2^{2^{m_1m_2}}$ for $\zeta$. 
\end{proof}

In an similar way we can show that the size of formula automata for 
other \secenl atomic formulae are also elementary.

\begin{lemma}
For any $\zeta \in \secenl$ the size of the automaton $\A(\zeta)$ for $\zeta$ is elementary.  
\end{lemma}

\section{Formalizing timing diagrams}
\label{section:td}
In this section we give a formal semantics to timing diagrams 
and formula translation from timing diagrams to \secenl. 
We begin by giving a textual syntax for timing diagrams 
which is derived from the timing diagram format of WaveDrom \cite{CP16,Wav16}. 

The symbols in a \emph{waveform} come 
from $\Lambda=\{\0,\1,\2,\x,\0|,\1|,\2|,\x|\}$ and 
$\Theta$, an atomic set of nominals. 
Let $\Gamma=\Theta\union\Lambda$. 
The syntax of a \emph{waveform} over $\Gamma$ is given by the grammar:
\[
\pi:=\0\ \ \|\ \ \1\ \ \|\ \ \2\ \ \|\ \ \x\ \ \|\ \ \0|\ \ \|\ \ \1|\ \ \|\ \ \2|\ \ \|\ \ \x|\ \ \|\ \ u:\pi\ \ \|\ \ \pi_1\pi_2, 
\]
where $u\in\Theta$ and $\pi\in\Lambda$. 
We call the elements in $\Theta$ the \emph{nominals}. 
As we shall see later, 
when we convert a waveform to a \secenl formula 
the nominals that appear in the formula are exactly 
the nominals in the waveform and hence the name. 
Let \wf be the set of all waveforms over $\Gamma$. 
 
An example of a waveform is 
\textsf{\0\1\textsf{a}:\2\x\0\1\1\x\textsf{b}:\x\2$|$\2\2\0\textsf{c}:\0\0} 
with $\Theta=\{\textsf{a,b,c}\}$. 
Intuitively, in a waveform \0 denotes \emph{low}, 
\1 \emph{high}, \2 and \x\ \emph{don't care}s 
(there is a subtle difference between \2 and \x\ though) 
and ``$|$'' the \emph{stuttering} operator. 


Let $\Sigma$ be a set of propositional variables. 
A \emph{timing diagram over $\Sigma$} is a tuple $\ang{\W,\Sigma,C,\Theta}$ 
where $\W=\{W_p\in\wf\ |\ p\in\Sigma\}$ and 
$C\subset\Theta\times\Theta\times\intv{\nat}$ 
a set of timing constraints.

Fig.~\ref{fig:wavedrom2} shows an example timing diagram 
$T=\ang{\{W_p,W_q\},\{(a,b,[10:10]), (a,$ 
$d,[1:8]), (c,d,[20:30])\}, \{a,b,c,d,e,f\}}$ 
along with its rendering in WaveDrom. 
The shared nominals have to be renamed in WaveDrom 
as commented in \S\ref{subsection:secenl}, 
in this case $a$ and $c$ in $W_q$ have been renamed $g$ and $h$ respectively. 
As in the case with \secenl formulas, 
nominals act as place holders in timing diagrams which 
can be shared among multiple waveforms. 
For example, in the figure $W_p$ and $W_q$ share the nominals $a$ and $c$. 
As a result a timing constraint in one timing diagram 
can implicitly induce a timing constraint in the other.   
For instance, even though there is no direct timing constraint between $a$ 
and $c$ in $W_p$ the constraints between $a$ and $d$, and $d$ and $c$ 
together impose one on them. 
\begin{figure}[!h]
\framebox{\parbox[t][11mm]{\textwidth}{
$\begin{array}{l}
\textsf{waveform \ensuremath{W_p} - \ \ensuremath{\mathrm{\0\1\textsf{a}:\2\x\0\1\1\x\textsf{b}:\x\2|\2\2\0\textsf{c}:\0\0}}}\\
\textsf{waveform \ensuremath{W_q} - \ \ensuremath{\mathrm{\0\0\textsf{a}:\0|\textsf{d}:\1\1|\textsf{e}:\x\x\x|\textsf{f}:\0\1\textsf{c}:\1\1}}}\\
\textsf{timing constraints: d-a$\in$[1:8], c-d$\in$[20:30], b-a$\in$[10:10]}\\
\end{array}$}}
\includegraphics[width=\textwidth, keepaspectratio]{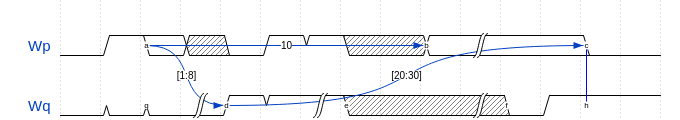}
\caption{Timing diagram $T$ and its WaveDrom rendering.}
\label{fig:wavedrom2}
\end{figure}

Let $T=\ang{\W,\Sigma,C,\Theta}$, $\W=\{W_p\in\wf\ |\ p\in\Sigma\}$, 
be a timing diagram. 
Let $\nu:\Theta\to[b,e]$ be a nominal valuation. 
Let $\sigma:[0,n]\to2^\Sigma$ be a word over $\Sigma$ and 
for all $p\in\Sigma$ let $\sigma_p:[0,n]\to\{0,1\}$ 
given by $\sigma_p(i)=1$ iff $p\in\sigma(i)$. 
Then the satisfaction relation $\sigma_p$ over 
a waveform $W$ under the valuation $\nu$ 
is defined as follows. 
\[
\begin{array}{rcl}
\sigma_p,[b,e]\modelsv \0&\mbox{iff}&e=b+1\mbox{ and }\sigma_p(b)=0,\\
\sigma_p,[b,e]\modelsv \1&\mbox{iff}&e=b+1\mbox{ and }\sigma_p(b)=1,\\
\sigma_p,[b,e]\modelsv \lambda&\mbox{iff}&e=b+1\mbox{ and }\lambda\in\{\2,\x\},\\
\sigma_p,[b,e]\modelsv \0|&\mbox{iff}&\forall b\leq i<e:\sigma_p(i)=0,\\
\sigma_p,[b,e]\modelsv \1|&\mbox{iff}&\forall b\leq i<e:\sigma_p(i)=1,\\
\sigma_p,[b,e]\modelsv \2|&\mbox{iff}&\forall b\leq i<e:\sigma_p(i)\in\{0,1\},\\
\sigma_p,[b,e]\modelsv \x|&\mbox{iff}&\forall b\leq i<e:\sigma_p(i)=1
\mbox{ or }\forall b\leq i<e:\sigma_p(i)=0,\\
\sigma_p,[b,e]\modelsv u:W&\mbox{iff}&\nu(u)=b\mbox{ and }\sigma_p,[b,e]\modelsv W,\\
\sigma_p,[b,e]\modelsv VW&\mbox{iff}&\exists b\leq i<e:\sigma_p,[b,i]\models_{\nu_1} V
\mbox{ and }\sigma_p,[i,e]\models_{\nu_2} W,\\
&&\mbox{ and }\nu_1||\nu\mbox{ and }\nu_2||\nu.\\
\end{array}
\]
We say $\nu\models C$ iff $\forall (a,b,\ang{l,r})\in C:\nu(b)-\nu(a)\in\ang{l,r}$.
We define $\sigma,[b,e]\modelsv\ang{\W,\Sigma,C,\Theta}$ iff 
$\forall p\in\Sigma:\sigma_p,[b,e]\modelsv W_p$ and $\nu\models C$.

\subsection{Waveform to \secenl translation}
\label{subsection:td2secenl}
We translate a waveform $W_p$ to \secenl as follows: 
every \0 occurring in $P$ is translated to \{\{$\neg$ \textsf{P}\}\},
\1 to \{\{\textsf{P}\}\},
\2 and \x\ to \textsf{slen=1},
\0$|$ to \textsf{pt$\vee$[$\neg$ P]}, 
\1$|$ to \textsf{pt$\vee$[P]}, 
\2$|$ to \textsf{true}, and 
\x$|$ to \textsf{pt$\vee$[P]$\vee$[$\neg$ P]}. 
A nominal $u$ that is appearing in $W_p$ 
is translated to \textsf{\textless u\textgreater}. 
For instance, the waveform 
$W_p$=\textsf{\0\1\textsf{a}:\2\x\0\1\1\x\textsf{b}:\x\2$|$\2\2\0\textsf{c}:\0\0} 
in $T$ of Fig.~\ref{fig:wavedrom2} will be translated to \secenl formula as below.
\begin{table}[!h]
\begin{tabular}{cccccccc}
\multicolumn{8}{l}{
\textsf{(\{\{$\neg$ P\}\}\textasciicircum\{\{P\}\}\textasciicircum
\textless a\textgreater\textasciicircum (slen=1)\textasciicircum (slen=1)\textasciicircum\{\{$\neg$ P\}\}\textasciicircum
\{\{P\}\}\textasciicircum\{\{P\}\}\textasciicircum(slen=1)\textasciicircum\textless b\textgreater\textasciicircum}\qquad\qquad}\\
&&&&\multicolumn{4}{r}{\textsf{(slen=1)\textasciicircum\textsf{true}\textasciicircum(slen=1)\textasciicircum(slen=1)
\textasciicircum\{\{$\neg$ P\}\}\textasciicircum\textless c\textgreater\textasciicircum\{\{$\neg$ P\}\}\textasciicircum\{\{$\neg$ P\}\}).}}\\
\end{tabular}
\end{table}

\noindent We denote the translated \secenl formula by $\xi(T,W_p)$. 
Similarly we can translate $W_q$ to get the formula $\xi(T,W_q)$. 
The timing constraints in $C$ is roughly translated to 
the \secenl formula $\xi(T,C)$ as follows.
\begin{table}[!h]
\begin{tabular}{cccccccc}
\multicolumn{8}{l}{
\textsf{((true\textasciicircum\textless a\textgreater\textasciicircum
((slen$\geq$ 1)\ $\wedge$\ (slen$\leq$ 8))\textasciicircum\textless d\textgreater\textasciicircum true)\ $\wedge$\qquad\qquad\qquad\qquad\qquad\qquad\qquad\qquad}}\\
\multicolumn{8}{c}{\textsf{(true\textasciicircum\textless d\textgreater\textasciicircum
((slen$\geq$ 20)\ $\wedge$\ (slen$\leq$ 30))\textasciicircum\textless c\textgreater\textasciicircum true)\ $\wedge$}}\\
\multicolumn{8}{r}{\textsf{(true\textasciicircum\textless a\textgreater\textasciicircum
(slen=10)\textasciicircum\textless b\textgreater\textasciicircum true)).}}\\
\end{tabular}
\end{table}

We define $\xi(T)=\xi(T,W_p)\And\xi(T,W_q)\And\xi(T,C)$. 
For a timing diagram $T=\ang{\W,\Sigma,C,\Theta}$, $\W=\{W_p\ |\ p\in\Sigma\}$ 
we define $\xi(T)=\bigwedge_{p\in\Sigma}\xi(T,W_p)\bigwedge\And\xi(T,C)$. 
\begin{theorem}
\label{lem:tdtosecenl} 
Let $T$ be a timing diagram. 
Then, for all $\sigma\in\Sigma^*$, for all $[b,e]\in\intv{\sigma}$ and 
for all nominal valuation $\nu$ over $[b,e]$, 
$\sigma,[b,e]\modelsv T$ iff $\sigma,[b,e]\modelsv\xi(T):\Theta$.  
Also, the translation $\xi(T):\Theta$ is linear in the size of $T$. 
\end{theorem}

\begin{proof}
Proof is not difficult and is by induction on the 
length of the waveform.
\end{proof}

Due above theorem we can now use timing diagrams in place of 
nominated formulas with liveness operators. 
We call such timing diagrams \emph{live timing diagrams}. 
For an example of a live timing diagram see Fig.~\ref{fig:lags}.

\subsection{Comparision with other temporal logics}
\label{subsection:comparision}

In previous section, Lemma \ref{lem:tdtosecenl} showed that timing diagrams can be
translated to equivalent  \sececntnl formulas with only linear blowup in size.
In this section we compare our logic SeCeNL
with other relevent logics in the literature viz, 
\ltl, discrete time \mtl, and \psl.  Of these, \psl is the most expressive and discrete time \mtl and \ltl are its syntactic subset. 
We show by  examples that \sececntnl formulae are more succint (smaller in size) than \psl and
we believe that they capture the diagrams more directly. 
Appendix \ref{section:more-examples} gives several more examples which 
could not be included due to lack of space.

\paragraph{Example (Ordered Stack)}
Let us now consider the timing diagram in Fig.~\ref{fig:eg1} adapted from \cite{CF05}.
Rise and fall of successive signals follow a stack discipline.
\begin{figure}
\centering
\includegraphics[width=.75\textwidth, keepaspectratio]{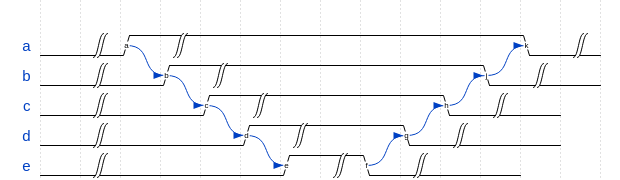}
\caption{Example 1.}
\label{fig:eg1}
\end{figure}
\noindent The language described by it is given by the \secenl formula: 
\[
\begin{small}
\begin{array}{l}
([\neg a]\ \textrm{\textasciicircum}\textless ua\textgreater\ \textrm{\textasciicircum}\ [a]\ \textrm{\textasciicircum}\ \textless va\textgreater\ \textrm{\textasciicircum}\ [\neg a])\ \wedge\ 
([\neg b]\ \textrm{\textasciicircum}\textless ub\textgreater\ \textrm{\textasciicircum}\ [b]\ \textrm{\textasciicircum}\ \textless vb\textgreater\ \textrm{\textasciicircum}\ [\neg b])\ \wedge\\ 
([\neg c]\ \textrm{\textasciicircum}\textless uc\textgreater\ \textrm{\textasciicircum}\ [c]\ \textrm{\textasciicircum}\ \textless vc\textgreater\ \textrm{\textasciicircum}\ [\neg c])\ \wedge\ 
([\neg d]\ \textrm{\textasciicircum}\textless ud\textgreater\ \textrm{\textasciicircum}\ [d]\ \textrm{\textasciicircum}\ \textless vd\textgreater\ \textrm{\textasciicircum}\ [\neg d])\ \wedge\\ 
([\neg e]\ \textrm{\textasciicircum}\textless ue\textgreater\ \textrm{\textasciicircum}\ [e]\ \textrm{\textasciicircum}\ \textless ve\textgreater\ \textrm{\textasciicircum}\ [\neg e])\ \wedge\  
(\ext\ \textrm{\textasciicircum}\ \textless ua\textgreater\ \textrm{\textasciicircum}\ \ext\ 
\textrm{\textasciicircum}\ \textless ub\textgreater\ \textrm{\textasciicircum}\ \ext)\ \And\\
(\ext\ \textrm{\textasciicircum}\ \textless ub\textgreater\ \textrm{\textasciicircum}\ \ext\ 
\textrm{\textasciicircum}\ \textless uc\textgreater\ \textrm{\textasciicircum}\ \ext)\ \And\
(\ext\ \textrm{\textasciicircum}\ \textless uc\textgreater\ \textrm{\textasciicircum}\ \ext\ 
\textrm{\textasciicircum}\ \textless ud\textgreater\ \textrm{\textasciicircum}\ \ext)\ \And\\
(\ext\ \textrm{\textasciicircum}\ \textless ud\textgreater\ \textrm{\textasciicircum}\ \ext\ 
\textrm{\textasciicircum}\ \textless ue\textgreater\ \textrm{\textasciicircum}\ \ext)\ \And\
(\ext\ \textrm{\textasciicircum}\ \textless va\textgreater\ \textrm{\textasciicircum}\ \ext\ 
\textrm{\textasciicircum}\ \textless vb\textgreater\ \textrm{\textasciicircum}\ \ext)\ \And\\
(\ext\ \textrm{\textasciicircum}\ \textless vb\textgreater\ \textrm{\textasciicircum}\ \ext\ 
\textrm{\textasciicircum}\ \textless vc\textgreater\ \textrm{\textasciicircum}\ \ext)\ \And\
(\ext\ \textrm{\textasciicircum}\ \textless vc\textgreater\ \textrm{\textasciicircum}\ \ext\ 
\textrm{\textasciicircum}\ \textless vd\textgreater\ \textrm{\textasciicircum}\ \ext)\ \And\\
(\ext\ \textrm{\textasciicircum}\ \textless vd\textgreater\ \textrm{\textasciicircum}\ \ext\ 
\textrm{\textasciicircum}\ \textless ve\textgreater\ \textrm{\textasciicircum}\ \ext). 
\end{array}
\end{small}
\] 
Note that first five conjuncts exactly correspond to the five waveforms. The last constraint enforces the ordering constraints
between waveforms. In general, if $n$ signals are stacked, its \sececntnl specification has size \bigo{n}. 

An equivalent \mtl (or \ltl) formula is given by:
\begin{small}
\[
\begin{array}{l}
[\neg a\ \And\ \neg b\ \And\ \neg c\ \And\ \neg d\ \And\ \neg e]\ \UU\ [a\ \And\ \neg b\ \And\ \neg c\ \And\ \neg d\ \And\ \neg e]\ \UU 
\\ \hspace*{0.5cm} 
[a\ \And\ b\ \And\ \neg c\ \And\ \neg d\ \And\ \neg e]\ \UU\ [a\ \And\ b\ \And\ c\ \And\ \neg d\ \And\ \neg e]\ \UU
\\ \hspace*{0.5cm} 
[a\ \And\ b\ \And\ c\ \And\ d\ \And\ \neg e]\ \UU\ [a\ \And\ b\ \And\ c\ \And\ d\ \And\ e]\ \UU
\\ \hspace*{0.5cm} 
[a\ \And\ b\ \And\ c\ \And\ d\ \And\ \neg e]\ \UU\ [a\ \And\ b\ \And\ c\ \And\ \neg d\ \And\ \neg e]\ \UU
\\ \hspace*{0.5cm} 
[a\ \And\ b\ \And\ \neg c\ \And\ \neg d\ \And\ \neg e]\ \UU\ [a\ \And\ \neg b\ \And\ \neg c\ \And\ \neg d\ \And\ \neg e]\ \UU
\\ \hspace*{0.5cm} 
[\neg a\ \And\ \neg b\ \And\ \neg c\ \And\ \neg d\ \And\ \neg e]
\end{array}
\]
\end{small} 
where $a\ \UU\ b$ is the derived modality $a\ \And\ \X(a\U b)$. For a stack of $n$ signals, the size of the \mtl formula is
$\bigo{\ensuremath{n^2}}$.

Above formula is also a \psl formula. We attempt to specify the pattern as a \psl regular expression as follows:\\
\begin{small}
\[
\begin{array}{l}
((\neg a\ \And\ \neg b\ \And\ \neg c\ \And\ \neg d\ \And\ \neg e; )[+]; ~
  (a\ \And\ \neg b\ \And\ \neg c\ \And\ \neg d\ \And\ \neg e; )[+]; 
  \\ \hspace*{0.5cm}
  (a\ \And\ b\ \And\ \neg c\ \And\ \neg d\ \And\ \neg e; )[+]; ~
  (a\ \And\ b\ \And\ c\ \And\ \neg d\ \And\ \neg e; )[+]; 
  \\ \hspace*{0.5cm}
  (a\ \And\ b\ \And\ c\ \And\ d\ \And\ \neg e; )[+]; ~
  (a\ \And\ b\ \And\ c\ \And\ d\ \And\ e; )[+];
   \\ \hspace*{0.5cm}
  (a\ \And\ b\ \And\ c\ \And\ d\ \And\ \neg e; )[+];~
  (a\ \And\ b\ \And\ c\ \And\ \neg d\ \And\ \neg e; )[+];
   \\ \hspace*{0.5cm}
  (a\ \And\ b\ \And\ \neg c\ \And\ \neg d\ \And\ \neg e; )[+];~
  (a\ \And\ \neg b\ \And\ \neg c\ \And\ \neg d\ \And\ \neg e; )[+];
   \\ \hspace*{0.5cm}
  (\neg a\ \And\ \neg b\ \And\ \neg c\ \And\ \neg d\ \And\ \neg e; )[+]
\end{array}
\]
\end{small}
For a stack of $n$ signals, the size of the \psl SERE expression is $\bigo{\ensuremath{n^2}}$. We believe that there is no formula of size
$\bigo{n}$ in \psl which can express the above property. Compare this with size $\bigo{n}$ formula of \sececntnl.

\paragraph{Example (Unordered Stack)}
In ordered stack signal $a$ turns on first and turns off last followed by signals $b,c,d,e$ in that order.
We consider a variation of the ordered stack example above where signals turn on and off in 
first-on-last-off order but there is no restriction on which signal becomes high first.
This can be compactly specified in \sececntnl as follows.
\[
\begin{small}
\begin{array}{l}
([\neg a]\ \textrm{\textasciicircum}\textless ua\textgreater\ \textrm{\textasciicircum}\ [a]\ \textrm{\textasciicircum}\ \textless va\textgreater\ \textrm{\textasciicircum}\ [\neg a])\ \wedge\ 
([\neg b]\ \textrm{\textasciicircum}\textless ub\textgreater\ \textrm{\textasciicircum}\ [b]\ \textrm{\textasciicircum}\ \textless vb\textgreater\ \textrm{\textasciicircum}\ [\neg b])\ \wedge\\ 
([\neg c]\ \textrm{\textasciicircum}\textless uc\textgreater\ \textrm{\textasciicircum}\ [c]\ \textrm{\textasciicircum}\ \textless vc\textgreater\ \textrm{\textasciicircum}\ [\neg c])\ \wedge\ 
([\neg d]\ \textrm{\textasciicircum}\textless ud\textgreater\ \textrm{\textasciicircum}\ [d]\ \textrm{\textasciicircum}\ \textless vd\textgreater\ \textrm{\textasciicircum}\ [\neg d])\ \wedge\\ 
([\neg e]\ \textrm{\textasciicircum}\textless ue\textgreater\ \textrm{\textasciicircum}\ [e]\ \textrm{\textasciicircum}\ \textless ve\textgreater\ \textrm{\textasciicircum}\ [\neg e])\ \wedge\ 
(\ext\ \textrm{\textasciicircum}\ \textless u1\textgreater\ \textrm{\textasciicircum}\ \ext\ 
\textrm{\textasciicircum}\ \textless u2\textgreater\ \textrm{\textasciicircum}\ \ext)\ \And\\ 
(\ext\ \textrm{\textasciicircum}\ \textless u2\textgreater\ \textrm{\textasciicircum}\ \ext\ 
\textrm{\textasciicircum}\ \textless u3\textgreater\ \textrm{\textasciicircum}\ \ext)\ \And\  
(\ext\ \textrm{\textasciicircum}\ \textless u3\textgreater\ \textrm{\textasciicircum}\ \ext\ 
\textrm{\textasciicircum}\ \textless u4\textgreater\ \textrm{\textasciicircum}\ \ext)\ \And\\
(\ext\ \textrm{\textasciicircum}\ \textless u4\textgreater\ \textrm{\textasciicircum}\ \ext\ 
\textrm{\textasciicircum}\ \textless u5\textgreater\ \textrm{\textasciicircum}\ \ext)\ \And\  
(\ext\ \textrm{\textasciicircum}\ \textless v5\textgreater\ \textrm{\textasciicircum}\ \ext\ 
\textrm{\textasciicircum}\ \textless v4\textgreater\ \textrm{\textasciicircum}\ \ext)\ \And\\ 
(\ext\ \textrm{\textasciicircum}\ \textless v4\textgreater\ \textrm{\textasciicircum}\ \ext\ 
\textrm{\textasciicircum}\ \textless v3\textgreater\ \textrm{\textasciicircum}\ \ext)\ \And\  
(\ext\ \textrm{\textasciicircum}\ \textless v3\textgreater\ \textrm{\textasciicircum}\ \ext\ 
\textrm{\textasciicircum}\ \textless v2\textgreater\ \textrm{\textasciicircum}\ \ext)\ \And\\
(\ext\ \textrm{\textasciicircum}\ \textless v2\textgreater\ \textrm{\textasciicircum}\ \ext\ 
\textrm{\textasciicircum}\ \textless v1\textgreater\ \textrm{\textasciicircum}\ \ext)\ \And\\ 
Bijection(ua,ub,uc,ud,ue,va,vb,vc,vd,ve,u1,u2,u3,u4,u5,v1,v2,v3,v4,v5) 
\end{array}
\end{small}
\] 
where formula $Bijection$ below states that there is one to one correspondence between positions marked by $ua,ub,uc,ud,ue,va,vb,vc,vd,de$
and positions marked by $u1,u2,u3,u4,u5,v1,v2,v3,v4,v5$. Moreover, it states that if $u_a$ maps to say $u_3$ than $v_a$ must map to $v_3$ and so on. 
\[
 \begin{array}{l}
 [[ (u1 \vee u2 \vee u3 \vee u4 \vee u5) \Leftrightarrow (ua \vee ub \vee uc \vee ud \vee ue)]] ~\land~
[[ \bigwedge_{1 \leq i,j \leq 5, i \not=j}  \neg (u_i \land u_j) ]] 
\\ \mbox{}
[[(v1 \vee v2 \vee v3 \vee v4 \vee v5) \Leftrightarrow (va \vee vb \vee vc \vee vd \vee ve)]]  ~\land~
[[ \bigwedge_{1 \leq i,j \leq 5, i \not=j}  \neg (v_i \& v_j) ]] 
\\ \mbox{}
\bigwedge_{1 \leq i \leq 5, j \in a,b,c,d,e}~ (true\ \textrm{\textasciicircum}\ \textless u_i \land u_j \textgreater\ \textrm{\textasciicircum}\ true \Leftrightarrow\ true\ \textrm{\textasciicircum}\ \textless v_i \land v_j \textgreater\ \textrm{\textasciicircum}\ true)
\end{array}
\]
Note that, in general, if $n$ signals are stacked,  then the above \sececntnl specification has size $\bigo{n^2}$.

Now we discuss encoding of unordered stack in \psl. In absence of nominals, it is difficult to state the above behaviour succinctly in logics \psl even using its SERE regular expressions.  Each order of occurrence of signals has to be enumerated as a disjunction where each
disjunct is as in the example ordered stack (where the order was $a, b, c, d, e$).
As there are $n!$ orders possible between $n$ signals, the size of the  \psl formula is also \bigo{n!}. We believe that there is no
polynomially sized formula in \psl encoding this property.
This shows that \sececntnl is exponentially more succint as compared to \psl.

In general, presence of nominals distinguishes \secenl from logics like \psl. In formalizing behaviour of hardware circuits
it has been proposed that regular expressions are not enough and operators such as pipelining have been introduced
\cite{CF05}. These are a form of synchronization and they can be easily expressed using nominals too. 
%

\section{Case study: Minepump Specification}
\label{section:casestudy}

We first specify some useful generic timing diagram properties which would 
used for requirement specification in this (and many other) case studies. 
\begin{itemize}
\item $\lags(P,Q,n)$: it is defined by Fig.~\ref{fig:lags2}. 
It specifies that in any observation interval 
if $P$ holds continuously for $n+1$ cycles and persists 
then $Q$ holds from $(n+1)^{th}$ cycle onwards 
and persists till $P$ persists. 
\item $\tracks(P,Q,n)$: defined Fig.~\ref{fig:tracks}. 
In any observation interval if $P$ 
becomes true then $Q$ sustains as long as $P$ sustains 
or upto $n$ cycles whichever is shorter. 
\item $\sep(P,n)$: Fig.~\ref{fig:separation} defines this property. 
Any interval which begins with a falling edge of $P$ and 
ends with a rising edge of $P$ then the length 
of the interval should be at least $n$ cycles.
\item $\ubound(P,n)$: Fig.~\ref{fig:ubound} defines the property.  
In any observation interval $P$ can be continuously true 
for at most $n$ cycles.
\end{itemize}
Note that we have presented these formulae diagrammatically. The textual version
of these live timing diagrams can be found in Appendix \ref{section:text-minepump}.
\begin{figure}[!h]
\begin{minipage}{.49\textwidth}
\centering
\includegraphics[height=21mm, keepaspectratio]{lags}
\caption{$\lags(P,Q,n)$.}
\label{fig:lags2}
\end{minipage}
\begin{minipage}{.49\textwidth}
\centering
\includegraphics[height=21mm,keepaspectratio]{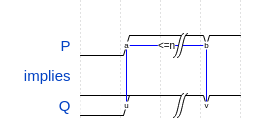}
\caption{$\tracks(P,Q,n)$.}
\label{fig:tracks}
\end{minipage}
\begin{minipage}{.49\textwidth}
\centering
\includegraphics[height=15mm, keepaspectratio]{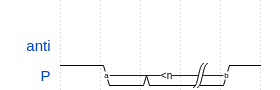}
\caption{$\sep(P,n)$.}
\label{fig:separation}
\end{minipage}
\begin{minipage}{.49\textwidth}
\centering
\includegraphics[height=15mm, keepaspectratio]{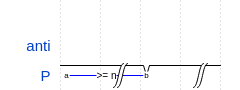}
\caption{$\ubound(P,n)$.}
\label{fig:ubound}
\end{minipage}
\end{figure}

We now state the minepump problem.
Imagine a minepump which keeps the water level in a mine 
under control for the safety of miners. 
The pump is driven by a controller which can switch it \emph{on} and \emph{off}. 
Mines are prone to methane leakage trapped underground 
which is highly flammable. So as a safety measure 
if a methane leakage is detected the controller is not allowed to 
switch on the pump under no circumstances. 

The controller has two input sensors - 
HH2O which becomes 1 when water level is high, 
and HCH4 which is 1 when there is a methane leakage; 
and can generate two output signals - 
ALARM which is set to 1 to sound/persist the alarm, 
and PUMPON which is set to 1 to switch on the pump. 
The objective of the controller is to \emph{safely} operate 
the pump and the alarm in such a way that 
the water level is never dangerous, 
indicated by the indicator variable DH2O, 
whenever certain assumptions hold. 
We have the following assumptions on the mine and the pump.
\begin{itemize}
\item[-] Sensor reliability assumption: $\pref{[[DH2O\implies HH2O]]}$. 
If HH2O is false then so is DH2O. 
\item[-] Water seepage assumptions: $\tracks(HH2O, DH2O, \kappa_1)$. 
The minimum no.~of cycles for water level 
to become dangerous once it becomes high is $\kappa_1$. 
\item[-] Pump capacity assumption: $\lags(PUMPON, \neg HH2O, \kappa_2)$. 
If pump is switched on for at least $\kappa_2+1$ cycles 
then water level will not be high after $\kappa_2$ cycles. 
\item[-] Methane release assumptions: $\sep(HCH4,\kappa_3)$ and $\ubound(HCH4,\kappa_4)$. 
The minimum separation between the two leaks of methane is $\kappa_3$ cycles 
and the methane leak cannot persist for more than $\kappa_4$ cycles. 
\item[-] Initial condition assumption: 
$\textbf{init}(\textless \neg HH2O\textgreater\wedge\textless\neg HCH4\textgreater, slen=0)$. 
Initially neither the water level is high nor there is a methane leakage. 
\end{itemize}
Let the conjunction of these \secenl formulas be denoted as $MINEASSUME$.

The commitments are:
\begin{itemize}
\item[-] Alarm control: 
$\lags(HH2O, ALARM, \kappa_5)$ and $\lags(HCH4, ALARM, \kappa_6)$ and 
$\lags(\neg HH2O\ \wedge\ \neg HCH4, \neg ALARM, \kappa_7)$. 
If the water level is dangerous 
then alarm will be high after $\kappa_5$ cycles and 
if there is a methane leakage then alarm will be high after $\kappa_6$ cycles. 
If neither the water level is dangerous nor there is a methane leakage 
then alarm should be off after $\kappa_7$  cycle. 
\item[-] Safety condition: $\pref{ [[ \neg DH2O\ \wedge\ (HCH4 $\implies$ \neg PUMPON)]] }$.
The water level should never become dangerous and 
whenever there is a methane leakage pump should be off. 
\end{itemize}
Let the conjunction of these commitments be denoted as $MINECOMMIT$. Then the requirement over
the minepump controller is given by the formula $MINEASSUME \Rightarrow MINECOMMIT$.
A textual version of this full minepump specification, which can be input to our tools is
given in Appendix \ref{section:text-minepump}. Note that the require consists of a mixture of 
timing diagram constraints (such as pump capacity assumption above) as well as \secenl formulas
(such as Safety condition above). 

We can automatically synthesize a controller  for the values say 
$\kappa_1=10$, $\kappa_2=2$, $\kappa_3=14$, $\kappa_4=2$, and 
$\kappa_5=\kappa_6=\kappa_7=1$.
The tool outputs a SCADE/SMV controller meeting the specification. 
A snapshot of SCADE code for the controller synthesized by 
DCSynthG for minepump can be found in Appendix~\ref{appendix:scade}. 
If the specification is not realizable we output an explanation. 

A second case study of synchronous bus arbiter specification can be found in Appendix.~\ref{section:arbiter}.
We can automatically synthesize a property monitor for such requirement and use it to model check a given
arbiter design; or we can directly synthesize a controller meeting the requirement. The appendix gives results of 
both these experiments.

\bibliographystyle{alpha}
\bibliography{rmmRef}

\newpage
\appendix

\section{Examples of Comparision with other logics}
\label{section:more-examples}

\paragraph{Example 1 (Ordering with timing)} 
Consider the timing diagram in Fig.~\ref{fig:eg}
which says that $a$ holds invariantly in the interval $[0,i]$ where $i \geq 1$, 
$b$ holds invariantly in the interval $[i,j]$, $j \geq i+1$, and $c$ holds at $j$ and $j \leq n$. 
\begin{figure}
\centering
\includegraphics[width=.4\textwidth, keepaspectratio]{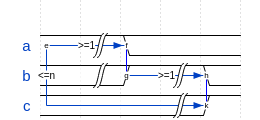}
\caption{Example 1.}
\label{fig:eg}
\end{figure}

\begin{itemize}
\item The language described by the above timing diagram is given by the \sececntnl formula  
$([a \And \neg  b]\ \textrm{\textasciicircum}\ [b \And \neg a \And \neg c]\ \textrm{\textasciicircum}\ 
\textless c\textgreater)\ \wedge\ (\textsf{slen} \leq n)$ which is of size $\bigo{\log(n)}$.
It is assumed that all timing constants such as $n$ are encoded in binary and hence they contribute size $\log(n)$.
 \item An equivalent \mtl formula is 
$\bigvee_{i=1}^{i=n-1}(a \And \neg b \U[i,i] (b \And \neg a \U[1,n-i]\ c))$ 
whose size is $\bigo{n\log(n)}$.
 \item Equivalent \ltl formula is 
$\bigvee_{i=1}^{i=n-1}\bigvee_{j=1}^{j=n-i}(a\ \U\X^i (b\ \U\X^j\ c))$ 
where $\X^k=\underbrace{\X\cdot\ldots\cdot\X}_{k\mbox{ times}}$, 
whose size is $\bigo{\ensuremath{n^2}}$.
\item Equivalent \psl formula is $(a \And \neg b[+]; ~b \And \neg a \And \neg c[+]; ~c) \And ((a|b)[<n];c)$ 
with size \bigo{\log(n)}.
\end{itemize}

We also give examples of complex dependancy constraints. Consider the timing diagram in Fig.~\ref{fig:eg2}. 
In this diagram, $ua$ occurs before $ub$ and $uc$, 
and $uc$ occurs before $ud$ and $ue$. 
The point $vc$ occurs after $vd$ and $ve$, 
and $va$ occurs after $vb$ and $vb$. 
\begin{figure}
\centering
\includegraphics[width=.75\textwidth, keepaspectratio]{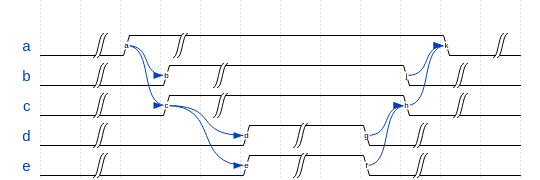}
\caption{Example 3.}
\label{fig:eg2}
\end{figure}

\noindent The behaviour is  described straightforwardly by the \sececntnl formula: 
\[
\begin{small}
\begin{array}{l}
([\neg a]\ \textrm{\textasciicircum}\textless ua\textgreater\textrm{\textasciicircum}\ [a]\ \textrm{\textasciicircum}\textless va\textgreater\textrm{\textasciicircum}\ [\neg a])\ \wedge\ ([\neg b]\ \textrm{\textasciicircum}\textless ub\textgreater\ \textrm{\textasciicircum}\ [b]\ \textrm{\textasciicircum} 
\textless vb\textgreater\textrm{\textasciicircum}[\neg b])\ \wedge\\
([\neg c]\ \textrm{\textasciicircum}\ \textless uc\textgreater\ \textrm{\textasciicircum}\ [c]\ \textrm{\textasciicircum}\ \textless vc\textgreater\ \textrm{\textasciicircum}\ [\neg c])\ \wedge\ ([\neg d]\ \textrm{\textasciicircum}\ \textless ud\textgreater\ \textrm{\textasciicircum}\ [d]\ \textrm{\textasciicircum}\textless vd\textgreater\ \textrm{\textasciicircum}\ [\neg d])\ \wedge\\
([\neg e]\ \textrm{\textasciicircum}\ \textless ue\textgreater\ \textrm{\textasciicircum}\ [e]\ \textrm{\textasciicircum}\ \textless ve\textgreater\ \textrm{\textasciicircum}\ [\neg e])\ \wedge\
(\ext\ \textrm{\textasciicircum}\ \textless ua\textgreater\ \textrm{\textasciicircum}\ ext\ \textrm{\textasciicircum}\ 
\textless ub\textgreater\ \textrm{\textasciicircum}\ true)\ \And\\ 
(\ext\ \textrm{\textasciicircum}\ \textless ua\textgreater\ \textrm{\textasciicircum}\ ext\ \textrm{\textasciicircum}\ 
\textless uc\textgreater\ \textrm{\textasciicircum}\ true)\ \And\ 
(\ext\ \textrm{\textasciicircum}\ \textless uc\textgreater\ \textrm{\textasciicircum}\ ext\ \textrm{\textasciicircum}\ 
\textless ud\textgreater\ \textrm{\textasciicircum}\ true)\ \And\\ 
(\ext\ \textrm{\textasciicircum}\ \textless uc\textgreater\ \textrm{\textasciicircum}\ ext\ \textrm{\textasciicircum}\ 
\textless ue\textgreater\ \textrm{\textasciicircum}\ true)\ \And\ 
(\ext\ \textrm{\textasciicircum}\ \textless ve\textgreater\ \textrm{\textasciicircum}\ ext\ \textrm{\textasciicircum}\ 
\textless vc\textgreater\ \textrm{\textasciicircum}\ true)\ \And\\ 
(\ext\ \textrm{\textasciicircum}\ \textless vd\textgreater\ \textrm{\textasciicircum}\ ext\ \textrm{\textasciicircum}\ 
\textless vc\textgreater\ \textrm{\textasciicircum}\ true)\ \And\ 
(\ext\ \textrm{\textasciicircum}\ \textless vc\textgreater\ \textrm{\textasciicircum}\ ext\ \textrm{\textasciicircum}\ 
\textless va\textgreater\ \textrm{\textasciicircum}\ true)\ \And\\ 
(\ext\ \textrm{\textasciicircum}\ \textless vb\textgreater\ \textrm{\textasciicircum}\ ext\ \textrm{\textasciicircum}\ 
\textless va\textgreater\ \textrm{\textasciicircum}\ true). 
\end{array}
\end{small}
\]
This formula is linear in the size of the timing diagram. Unfortunately, specifying these dependancies in \psl is complex and
formula size blows up at least quadratically.

\section{Implementation}
\label{section:xml}
We propose a textual framework 
with a well defined syntax and semantics 
for requirement specification 
(of the form assumptions \implies\ commitments). 
Our framework is \emph{heterogeneous} in the sense that 
it supports both \secenl formulas and 
timing diagrams with nominals for system specification. 
It can also handle all of our limited liveness operators. 
(see Appendix.~\ref{section:text-minepump} for the code 
for minepump in our framework). 

We have also developed a \emph{Python based translator} which takes 
requirements in our textual format as input 
and produces property monitors as well as controllers as output. 
Fig.~\ref{fig:tool-chain} gives a 
broad picture of the current status of our tool chain.

\begin{figure}[!h]
\begin{center}
\begin{tikzpicture}[scale=0.28, every node/.style={scale=0.75}]
\draw (1,2) node {requirement} (1,1) node {specification:} (1,0) node {timing diagrams} (1,-1) node {+\secenl} (1,-2) node {+liveness};
\draw[->] (4.5,0) -- (8,0);
\draw (8,-2) rectangle (12,2);
\draw (10,.5) node {Python} (10,-0.5) node {translator};
\draw[->] (12,0) -- (16,0) node[midway,above] {QDDC};
\draw (16,-3.5) rectangle (25,3.5);
\draw (20,-1.5) rectangle (25,1.5);
\draw (22.5,0) node {DCVALID} (18.5,2.5) node {DCObs} (18.5,-2.5) node {DCSynthG};
\draw (16,0) -- (20,0);
\draw[->] (37,3) -- (37,2);
\draw (37,3.5) node {system};
\draw (25,1) -- (34,1) node[midway,above] {property monitors};
\draw[->] (34,1) -- (35,1);
\draw (35,0) rectangle (39,2) (37,1.5) node {Model} (37,0.5) node {checker};
\draw (25,-1) -- (34,-1) node[midway,below] {synthesized controller};
\draw[->] (34,-1) -- (35,-1);
\draw[dashed] (-3,-4.5) rectangle (34,4.5); 
\end{tikzpicture}
\end{center}
\caption{Our tool chain.}
\label{fig:tool-chain}
\end{figure}
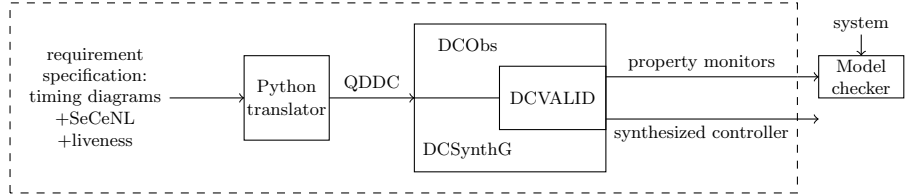

\section{Minepump Code}
\label{section:text-minepump} 
The example code for minepump is written using textual syntax for
QDDC which can be found in \cite{Pan00,Pan01}.

\noindent
\#lhrs "minepump"\\
interface\\
\{

input HH2O, HCH4; 

output ALARM monitor x, PUMPON monitor x;

constant delta = 1, w = 10, epsilon=2 , zeta=14, kappa=2;

auxvar DH2O;

softreq (!YHCH4)$||$(!PUMPON);\\
\}\\
\#implies lag(P, Q, n)\\
\{

td lagspeclet1(P, n)

\{

P:\textless u\textgreater1$|$\textless v\textgreater1$|$;

@sync:(u, v, n);

\}

td lagspeclet2(Q)

\{

Q: 2$|$\textless v\textgreater1$|$;

\}\\
\}\\
\#implies tracks(P, Q, n)\\
\{

td tracksspeclet1(P, n)

\{

P: 0\textless u\textgreater 1$|$\textless v\textgreater 1$|$;

@sync: (u,v,[n,));

\}

td tracksspeclet2(Q)

Q: 2\textless u\textgreater 1$|$\textless v\textgreater 0$|$;\\
\}\\
\#implies tracks2(P, Q, n)\\
{

td tracks2speclet1(P, n)

\{

P: 0\textless u\textgreater 1$|$\textless v\textgreater 0$|$;

@sync: (u,v,[,n]);

\}

td tracks2speclet2(Q)

Q:2\textless u\textgreater 0$|$\textless v\textgreater 2$|$;\\
\}\\
\#implies sep(P, n)\\
\{

td sepspeclet1(P)

P: 1\textless u\textgreater 0$|$\textless v\textgreater 1;

td sepspeclet2(n)

\{

@null: 2\textless u\textgreater 2$|$\textless v\textgreater 2;

@sync: (u, v, (n,]);

\}\\
\}\\
\#implies ubound(P, n)\\
\{

td boundspeclet1(P)

P: \textless c\textgreater 1$|$\textless d\textgreater 1;

td boundspeclet2(n)

\{

@null: \textless c\textgreater 2$|$\textless d\textgreater 2;

@sync: (c, d, [,n));

\}\\
\}\\
dc safe(DH2O)
\{

pt $||$ [!DH2O \&\& ((HCH4$||$!HH2O) =\textgreater  !PUMPON)];\\
\}\\
main()\\
\{

assume (\textless !HH2O\textgreater\ \textrm{\textasciicircum}\ true);

assume (pt $||$ [DH2O =\textgreater  HH2O]);

assume tracks(HH2O, !DH2O, w);

assume tracks2	(HH2O, DH2O, w);

assume lag(PUMPON, !HH2O, epsilon);

assume sep(HCH4, zeta);

assume ubound(HCH4, kappa);

req (\textless !ALARM\textgreater\ \textrm{\textasciicircum}\ true);

req lag(HH2O, ALARM, delta);

req lag(HCH4, ALARM, delta);

req lag(!HCH4 \&\& !HH2O, !ALARM, delta);

req safe(DH2O);\\
\}

\pagebreak

\section{Synthesized controller for minepump}
\label{appendix:scade}
A snapshot of a controller synthesized from the minepump requirement in 
\S\ref{section:casestudy}. The controller had approximately 140 states 
and it took less than a second for synthesis.

\begin{figure}[!h]
\framebox{\parbox[t][112mm]{\textwidth}{
$\begin{array}{l}

\mathrm{\textsf{node minepump ( HH2O, HCH4:bool) returns ( ALARM, PUMPON:bool)}}\\
\mathrm{\textsf{var cstate: int; }}\\
\mathrm{\textsf{let }} \\
 \mathrm{\textsf{\qquad ALARM, PUMPON, cstate = }}\\
 \mathrm{\textsf{\qquad ( if true  and not HH2O and not HCH4 then ( false,  false,   2) }}\\
  \mathrm{\textsf{\qquad else if true  and not HH2O and HCH4 then ( false,  false,   4) }}\\
  \mathrm{\textsf{\qquad else if true  and HH2O and not HCH4 then ( false,  false,   4) }}\\
  \mathrm{\textsf{\qquad else if true  and HH2O and HCH4 then ( false,  false,   4) }}\\
 \mathrm{\textsf{\qquad else ( dontCare,  dontCare,  1)) \implies }}\\
 \mathrm{\textsf{\qquad if pre cstate = 1 and not HH2O and not HCH4 then ( false,  false, 2) }}\\
 \mathrm{\textsf{\qquad else if pre cstate = 1 and not HH2O and HCH4 then ( false,  false, 4) }}\\
 \mathrm{\textsf{\qquad else if pre cstate = 1 and HH2O and not HCH4 then ( false,  false, 4) }}\\
 \mathrm{\textsf{\qquad else if pre cstate = 1 and HH2O and HCH4 then ( false,  false, 4) }}\\
 \mathrm{\textsf{\qquad else if pre cstate = 2 and not HH2O and not HCH4 then ( false,  false, 2) }}\\
 \mathrm{\textsf{\qquad else if pre cstate = 2 and not HH2O and HCH4 then ( false,  false, 7) }}\\
 \mathrm{\textsf{\qquad else if pre cstate = 2 and HH2O and not HCH4 then ( false,  false, 9) }}\\
 \mathrm{\textsf{\qquad else if pre cstate = 2 and HH2O and HCH4 then ( false,  false, 11) }}\\
 \mathrm{\textsf{\qquad else if pre cstate = 4 and not HH2O and not HCH4 then ( false,  false, 4) }}\\
 \mathrm{\textsf{\qquad else if pre cstate = 4 and not HH2O and HCH4 then ( false,  false, 4) }}\\
 \mathrm{\textsf{\qquad ..............................}}\\
 \mathrm{\textsf{\qquad ..............................}}\\
 \mathrm{\textsf{\qquad else if pre cstate = 309 and HH2O and not HCH4 then ( false,  true, 4) }}\\
 \mathrm{\textsf{\qquad else if pre cstate = 309 and HH2O and HCH4 then ( false,  true, 4) }}\\
 \mathrm{\textsf{\qquad else if pre cstate = 372 and not HH2O and not HCH4 then ( false,  false, 2) }}\\
 \mathrm{\textsf{\qquad else if pre cstate = 372 and not HH2O and HCH4 then ( false,  false, 7) }}\\
 \mathrm{\textsf{\qquad else if pre cstate = 372 and HH2O and not HCH4 then ( false,  true, 4) }}\\
 \mathrm{\textsf{\qquad else if pre cstate = 372 and HH2O and HCH4 then ( false,  true, 4) }}\\
 \mathrm{\textsf{\qquad else ( dontCare, dontCare, pre cstate) ; }}\\

\mathrm{\textsf{tel }}\\

\end{array}$}}
\end{figure}

\section{Case study: 3-cell arbiter}
\label{section:arbiter}
In this section we illustrate another application 
of our specification format and associated tools. 
For this we use the standard McMillan arbiter circuit 
given in NuSMV examples and do the model checking against the specification below.

A \emph{synchronous 3-cell bus arbiter} has 3 request lines 
$req1, req2$ and $req3$, and 
corresponding acknowledgement lines $ack1, ack2$ and $ack3$. 
At any time instance a subset of request lines can be high and 
arbiter decides which request should be granted permission 
to access the bus by making corresponding acknowledgement line high.
The requirements for such a bus arbiter are as formulated below.
\begin{itemize}
\item[-] Exclusion: $\pref{[[(\bigwedge_{i\neq j}\neg (ack_i\And ack_j))]]}$. 
At most 1 acknowledgement can be given at a time.
\item[-] No spurious acknowledgement: $\pref{[[(\bigwedge_1(ack_i\implies req_i))]]}$. 
A request should be granted access 
to the bus only if it has requested it.
\item[-] Response time: $\textbf{implies}([[req]]\And slen=n, \true\verb|^|\textless ack\textgreater\verb|^|\true)$.
One of the most important property of an arbiter 
is that it any request should be granted within $n$ cycles, 
i.~e.~if a request is continuously true for sometime then it should be heard. 
\item [-] Deadtime: to specify this property we first specify 
lost cycle as follows: $Lost\equiv(\bigvee_i req_i)\And(\neg (\bigvee_i ack_i))$. 
Then $Deadtime\equiv\textbf{anti}([[Lost]]\And slen>n)$. 
This specifies the maximum number of consecutive cycles 
that can be lost by the arbiter is $n$.
\end{itemize}
The requirement $ARBREQ$ is a conjunction of above formulas.

We ran the requirement through our tool chain 
to generate NuSMV module for the requirement monitor.
This  module was then instantiated synchronously with 
McMillan arbiter implementation in NuSMV 
and NuSMV model checker was called in to check 
the property $G(assumptions\implies commitments)$.

\paragraph{Model checking}: 
Experimental results show that the deadtime for 3-cell McMillan arbiter is 3. 
If we specify the deadtime as 2 cycles 
then a counter example is generated by NuSMV as depicted in Fig.~\ref{fig:deadtime1}. 
This counter examples show that even though there is an request line high in 
$4^{th}$, $5^{th}$ and $6^{th}$ cycle, but no acknowledgment is given by arbiter.
Similarly, the response time for $1^{st}$ request is 3 cycles 
whereas for $2^{nd}$ and $3^{rd}$ cell it is 6 cycles. 
If we specify the response time of 2 and 5 cycles for $1^{st}$ 
and $2^{nd}$ then NuSMV generates counter examples in 
Fig.~\ref{fig:resptime2} and Fig.~\ref{fig:resptime5} respectively. 
Fig.~\ref{fig:resptime5} shows that the request line for cell 2 
(i.~e.~req2) is high continuously for 5 cycles 
starting from $3^{rd}$ without an acknowledgement from the arbiter. 
\begin{figure}[!h]
\begin{minipage}{.32\textwidth}
\includegraphics[height=2.5cm, keepaspectratio]{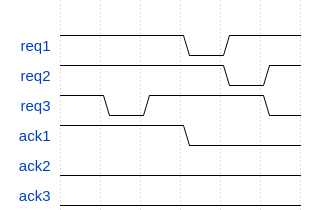}
\caption{Counter Example Showing deadtime exceeding 2 cycles}
\label{fig:deadtime1}
\end{minipage}
%
\begin{minipage}{.32\textwidth}
\includegraphics[height=2.5cm, keepaspectratio]{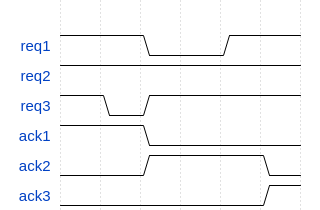}
\caption{Counter Example showing response time of 1st cell exceeding 2 cycles}
\label{fig:resptime2}
\end{minipage}
%
\begin{minipage}{.32\textwidth}
\includegraphics[height=2.5cm, keepaspectratio]{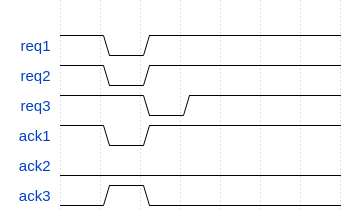}
\caption{Counter Example showing response time of 2nd cell exceeding 5 cycles}
\label{fig:resptime5}
\end{minipage}
\end{figure}

\paragraph{Controller synthesis}: 
We have also synthesized a controller for the arbiter specification using our tool DCSynthG.
We have tightened the requirements by specifying the response time as 3 cycles 
uniformly for all three cells and deadtime as 0 cycles, 
i.~e.~there is no lost cycle. 
The tool could synthesize a controller in 0.03 seconds with 17 states.

\end{document}